%% file: arxiv2.tex
\documentclass[runningheads,11pt]{llncs}

\usepackage{a4wide}
\usepackage{xspace}
\usepackage{algorithm}
\usepackage[noend]{algorithmic}
\usepackage{amsfonts,amssymb, amsmath}
\usepackage{xcolor}
\usepackage{graphicx}

\newtheorem{observation}[theorem]{Observation}

\newcommand{\floor}[1]{\lfloor #1 \rfloor}
\newcommand{\nats}{\ensuremath{\mathbb{N}}\xspace}

\newcommand{\opt}{\ensuremath{\textsc{OPT}}\xspace}
\newcommand{\algTwoBatches}{\ensuremath{\textsc{TwoBatches}}\xspace}
\newcommand{\height}{\ensuremath{h}\xspace}
\newcommand{\pdepth}{\ensuremath{d_p}\xspace}
\newcommand{\optcolor}{\ensuremath{\textsc{Color}}\xspace}
\newcommand{\batch}{\ensuremath{\textsc{BatchColor}_f}\xspace}
\newcommand{\kbatch}{\ensuremath{\textsc{$k$-BatchColor}}\xspace}
\newcommand{\alg}{\ensuremath{A}\xspace}

\newcommand{\push}[1]{\ensuremath{\text{Push}(#1)}\xspace}
\newcommand{\pop}{\ensuremath{\text{Pop}()}\xspace}

\newcommand{\icolor}[1]{\ensuremath{\text{color}(#1)}\xspace}
\newcommand{\colorset}{\ensuremath{\mathcal{C}}\xspace}
\newcommand{\setofintervals}{\ensuremath{\mathcal{I}}\xspace}
\newcommand{\leftordering}[2]{\ensuremath{L_{#1}(#2)}\xspace}
\newcommand{\rightordering}[2]{\ensuremath{R_{#1}(#2)}\xspace}

\newcommand{\totalorder}{\ensuremath{T}\xspace}
\newcommand{\smaller}{\ensuremath{<_T}\xspace}

\newcommand{\coloredBatchTwo}{\ensuremath{\textsc{Batch}_2\textsc{-Colored}}\xspace}

\newcommand{\SP}{\ensuremath{{\mathcal P}}\xspace}
\newcommand{\rSP}{\ensuremath{{\mathcal P}_R}\xspace}

\newcommand{\chainOne}{\ensuremath{\textsc{Chain}_1}\xspace}
\newcommand{\chainTwo}{\ensuremath{\textsc{Chain}_2}\xspace}
\newcommand{\chainThree}{\ensuremath{\textsc{Chain}_3}\xspace}
\newcommand{\chain}{\ensuremath{\textsc{Chain}}\xspace}

\newcommand{\rchains}{\ensuremath{\textsc{Chains}_R}\xspace}

\newcommand{\rchain}{\ensuremath{\textsc{Chain}_R}\xspace}
\newcommand{\ileft}{\ensuremath{I_{\ell}}\xspace}
\newcommand{\iright}{\ensuremath{I_{r}}\xspace}

\newcommand{\SET}[1]{\left\{#1\right\}}
\newcommand{\SETOF}[2]{\SET{#1 \mid #2}}
\newcommand{\FLOOR}[1]{\left\lfloor#1\right\rfloor}

\newcommand{\ratio}{\ensuremath{c}\xspace}
\newcommand{\largerratio}{\ensuremath{C}\xspace}

\begin{document}

\title{Batch Coloring of Graphs}
\titlerunning{Batch Coloring of Graphs} 

\author{Joan Boyar \inst{1}
        \and Leah Epstein \inst{2}
        \and Lene M. Favrholdt \inst{1}
        \and \\ Kim S. Larsen \inst{1}
        \and Asaf Levin \inst{3}}
\authorrunning{J. Boyar, L. Epstein, L.\,M. Favrholdt, K.\,S. Larsen, A. Levin} 
\institute{Dept. of Mathematics and Computer Science, University of
Southern Denmark, Odense, Denmark, \email{ \{joan,lenem,kslarsen\}@imada.sdu.dk}
\thanks{Supported in part by the Danish Council for Independent Research, Natural Sciences, grant DFF-1323-00247, and the Villum Foundation, grant VKR023219.} 
\and {Dept. of Mathematics, University of Haifa,
Haifa, Israel, \email{lea@math.haifa.ac.il}} \and {Faculty of
IE\&M, The Technion, Haifa,
Israel, \email{levinas@ie.technion.ac.il}}}



\maketitle

\begin{abstract}
In graph coloring problems, the goal is to assign a positive
integer color to each vertex of an input graph such that adjacent
vertices do not receive the same color assignment. For classic
graph coloring, the goal is to minimize the maximum color used,
and for the sum coloring problem, the goal is to minimize the sum
of colors assigned to all input vertices. In the offline variant,
the entire graph is presented at once, and in online problems, one
vertex is presented for coloring at each time, and the only
information is the identity of its neighbors among previously
known vertices. In batched graph coloring, vertices are presented
in $k$ batches, for a fixed integer $k \geq 2$, such that the
vertices of a batch are presented as a set, and must be colored
before the vertices of the next batch are presented. This last
model is an intermediate model, which bridges between the two
extreme scenarios of the online and offline models. We provide several
results, including a general result for sum coloring and results for
the classic graph coloring problem on restricted graph classes: We
show tight bounds for any graph class containing trees as a subclass
(e.g., forests, bipartite graphs, planar graphs, and perfect graphs),
and a surprising result for interval graphs and $k=2$,
where the value of the (strict and asymptotic) competitive ratio
depends on whether the graph is presented with its interval
representation or not.
\end{abstract}

\section{Introduction}

We study three different graph coloring problems in a model where
the input is given in {\em batches}.  In this model of computation an
adversary reveals the input graph one batch at a time.  Each batch
is a subset of the vertex set together with its edges to the
vertices revealed in the current batch or in previous batches.
After a batch is revealed the algorithm is asked to color the
vertices of this batch with colors which are positive integers,
the coloring must be valid or {\em proper}, i.e., neighbors are colored using
distinct colors, and this coloring cannot be modified later.

The batch scenario is somewhere between online and
offline. In an {\em offline} problem, there is only one batch,
while for an {\em online}
problem, the requests arrive one at a time and have to be handled
as they arrive without any knowledge of future events, so each
request is a separate batch. Many applications might fall between
these two extremes of online and offline. For example, a situation
where there are two (or more) deadlines, an early one with a lower price and a
later one with a higher price can lead to batches.

When considering a combinatorial
problem using batches, we assume that the requests arrive grouped
into a constant number $k$ of batches. Each batch must be handled
without any knowledge of the requests in future batches. As with
online problems, we do not consider the execution times of the
algorithms used within one batch; the focus is on the performance
ratios attainable. Therefore, our goal is to quantify the extent
to which the performance of the solution deteriorates due to the
lack of information regarding the requests of future batches. We
also investigate how much advance knowledge of the number of batches
can help.

The quality of the algorithms is evaluated using competitive
analysis. Let $\alg(\sigma)$ denote the cost of the solution returned
by algorithm $\alg$ on request sequence $\sigma$, and let $\opt(\sigma)$
denote the cost of an optimal (offline) solution. Note that for
standard coloring problems, $\opt(G) = \chi(G)$, where $\chi(G)$
is the chromatic number of the graph $G$.
An online coloring algorithm $\alg$ is \emph{$\rho$-competitive}
if there exists a constant $b$ such that, for all finite request
sequences $\sigma$, $\alg(\sigma)\leq \rho\cdot \opt(\sigma)+b$. The
\emph{competitive ratio} of algorithm $\alg$ is $\inf\{\rho \mid
\alg~ \mbox{\rm is } \rho\mbox{\rm -competitive}\}$. If the
inequality holds with $b=0$, the algorithm is \emph{strictly
$\rho$-competitive} and the {\em strict competitive ratio} is
$\inf\{ \rho \mid \alg~ \mbox{\rm is strictly } \rho
\mbox{\rm -competitive}\}$. 

The First-Fit algorithm for coloring a graph traverses the list of
vertices given in an arbitrary order or in the order they are
presented, and assigns each vertex the minimal color not assigned
to its neighbors that appear before it in the list of vertices.

Other combinatorial problems have been studied previously using
batches. The study of bin packing with batches was motivated by
the property that all known lower bound instances have the form that items
are presented in batches. The case
of two batches was first considered in~\cite{GJY05}, an algorithm
for this case was presented in~\cite{D15}, and better lower bounds
were found in~\cite{BBGM09}. A study of the more general case of $k$ batches was done in~\cite{E16}, and recently, a new
lower bound on the competitive ratio of bin packing with three
batches was presented in~\cite{BBDGT15}. The scheduling
problem of minimizing makespan on identical machines where jobs
are presented using two batches was considered in~\cite{ZCW03}.

\paragraph{Graph classes containing trees.}
The first coloring problem we consider using batches is that of
coloring graph classes containing trees as a subclass
(e.g., forests, bipartite graphs, planar graphs, perfect graphs, and graphs in
general),
minimizing the number of colors used.
Offline, finding a proper coloring of bipartite graphs is
elementary and only (at most) two colors are needed. However,
there is no online algorithm with a constant competitive ratio, even for trees.
Gy\'{a}rf\'{a}s and Lehel~\cite{GLff88} show that for any online
tree coloring algorithm $\alg$ and any $n \geq 1$, there is a tree
on $n$ vertices for which $\alg$ uses at least $\lfloor \log
n\rfloor+1$ colors. 
The lower bound is exactly matched by First-Fit~\cite{GL90}, and hence, the
optimal competitive ratio on trees is $\Theta(\log n)$. 
For general graphs, 
Halld\'{o}rsson and Szegedy~\cite{HS94} have shown that the
competitive ratio 
is $\Omega(n/\log n)$.

We show that any algorithm for coloring trees in $k$ batches uses
at least $2k$ colors in the worst case, even if the number of
batches is known in advance.
This gives a lower bound of $k$ on the competitive ratio of any
algorithm coloring trees in $k$ batches. 
The lower bound is tight, since (on any graph, not only trees),
a $k$-competitive algorithm can be obtained by coloring each
batch optimally with colors not used in previous batches.
Thus, for graph classes containing trees as a subclass,
$k$ is the optimal competitive ratio.

\paragraph{Coloring interval graphs in two batches.}
Next we consider coloring interval graphs in two batches, 
minimizing the number of colors used. An interval graph is a graph
which can be defined as follows: The vertices represent intervals
on the real line, and two vertices are adjacent if and only if
their intervals overlap (have a nonempty intersection). If the
maximum clique size of an interval graph is $\omega$, it can be
colored optimally using $\omega$ colors by using First-Fit on the
interval representation of the graph, with the intervals sorted by
nondecreasing left endpoints. For the online version of the
problem, Kierstead and Trotter~\cite{KT81} provided an algorithm
which uses at most $3\omega-2$ colors and proved a matching lower
bound for any online algorithm. 

The algorithm presented
in~\cite{KT81} does not depend on the interval representation of
the graph, but the lower bound does, so in the online case the
optimal competitive ratio is the same for these two
representations (see~\cite{KST16,NB08} for the current best
results regarding the strict competitive ratio of First-Fit for
coloring interval graphs).
In contrast, when there are two batches, there is a difference. We
show tight upper and lower bounds of $2$ for the case when
the interval representation is unknown and $3/2$ when it is
known,
respectively. Our results apply to both the asymptotic and the strict
competitive ratio.

Note that when the interval representation of the graph is used,
the batches consist of intervals on the real line (and it is not
necessary to give the edges explicitly).

\paragraph{Sum coloring.}
The sum coloring problem (also called chromatic sum) was
introduced in~\cite{KS89} (see~\cite{kub04} for a survey of
results on this problem). The problem is to give a proper coloring
to the vertices of a graph, where the colors are positive
integers, so as to minimize the sum of these colors over all
vertices (that is, if the coloring is defined by a function $C$,
the objective is to minimize $\sum_{v\in V} C(v)$).

Bar-Noy et~al.~\cite{BBHST98} study the problem, motivated by the
following application: Consider a scheduling problem on an
infinite capacity batched machine where all jobs have unit
processing time, but some jobs cannot be run simultaneously due to
conflicts for resources. If the conflicts are given by a graph
where the jobs are vertices and an edge exists between two
vertices, if the corresponding jobs cannot be executed
simultaneously (and thus each batch of jobs corresponds to an
independent set of this graph), the value $s$ of the optimal sum
coloring of the graph gives the sum of the completion times of all
jobs in an optimal schedule. Dividing by the number of jobs gives
the average response time.
The problem when restricted to interval graphs is also motivated by
VLSI routing~\cite{NSS99}. The first problem seems more likely to come
in batches than the second.

The sum coloring problem is NP-hard for general
graphs~\cite{KS89} and cannot be approximated within
$n^{1-\varepsilon}$ for any $\varepsilon>0$
unless $\mbox{ZPP}=\mbox{NP}$~\cite{BBHST98}.
Interestingly, there is a linear time algorithm for
trees, even though there is no constant upper bound on the number
of different colors needed for the minimum sum coloring of
trees~\cite{KS89}. For online algorithms, there is a lower
bound of $\Omega(n/\log^2 n)$ for general graphs
with $n$ vertices~\cite{H00}.

We show tight upper and lower bounds of $k$ on the competitive
ratio when there are $k$ batches and $k$ is known in advance to
the algorithm. The competitive ratio is higher if $k$ is unknown
in advance to the algorithm. We do not give a closed form
expression for the competitive ratio in this case, but give tight
upper and lower bounds on the order of growth of the competitive
ratio and the strict competitive ratio. For any nondecreasing
function $f$, with $f(1)\geq 1$, the optimal competitive ratio for
$k$ batches is $O(f(k))$ if the series $\sum_{i=1}^\infty
\frac{1}{f(i)}$ converges, and it is $\Omega(f(k))$ if the series
diverges. Thus, for example, it is $O(k \log k (\log\log k)^2)$
and $\Omega(k \log k \log\log k)$.

Restricting to trees, First-Fit is strictly $2$-competitive for
the online problem. Thus, First-Fit gives a (strict) competitive
ratio of $2$ regardless of the number of batches.
See for example~\cite{Bor+} for results on the strict competitive ratio of
First-Fit for other graph classes.

\section{Graph Classes Containing Trees}
In this section, we study the problem of coloring trees in $k$ batches.
The results hold for any graph class that
contains trees as a special case, including bipartite graphs,
planar graphs, perfect graphs, and the class of all graphs. 
If we want the algorithm to be polynomial time, then we are
restricted to graph classes where optimal offline coloring is
possible in polynomial time (e.g., perfect graphs~\cite{GLS81}).

The construction proving the following lemma resembles that of the
lower bound of $\Omega(\log n)$ for the competitive ratio for
online coloring of trees~\cite{GLff88}.

\begin{lemma}
\label{lbtreep}
For any integer $k \geq 1$, any algorithm for $k$-batch coloring
of trees can be forced to use at least $2k$
colors, even if $k$ is known in advance.
\end{lemma}
\begin{proof}
After each batch, the graph will be a forest. If at some point, at
least $2k+1$ distinct colors are used, no further batches will be
introduced (that is, any remaining batches will be empty). Thus,
in the discussion for each batch, we assume that the vertices of
the batch are colored with at most $2k$ colors.

Batch $i$, $1 \leq i \leq k$, contains $a_i = 2 \cdot
(8k^3)^{k-i}$ vertices. After batch $i$, the graph will contain
$a_i/2$ disjoint {\em level
  $i$ trees}.
A level $i$ tree consists of an edge, called the {\em base edge},
with each endpoint connected to two {\em good} level $j$ trees,
for $1 \leq j \leq i-1$. A good level $j$ tree is a level $j$ tree
with at least one vertex of each color $1,2,\ldots,2j$. Thus, a
good level 1 tree is just one edge, with colors $1$ and $2$ on its
endpoints.

We now explain how the batches are constructed. In particular, we
prove that there are enough good level $j$ trees, for each
$j$, $1 \leq j \leq k$. Once this is in place, we have proven
that, after the $k$th batch, each color $1,2,\ldots,2k$ has been
used.

The first batch is a matching over its vertices. Each edge of the
first batch must be colored with two different colors. Since we
are assuming that at most $2k$ colors are used, there are less
than $2k^2$ distinct pairs of colors. Thus, there are at least
$g_1 = a_1/(4k^2)$ edges having the same pair of colors. We rename
the colors such that these are colors $1$ and $2$. These edges are
then the good level $1$ trees and these are the base edges.

In batch $i \geq 2$, there are $a_i/2$ new base edges.
Furthermore, for each $j$, $1 \leq j \leq i-1$, each of the $a_i$
new vertices is connected to a vertex of color $2j-1$ in a good
level $j$ tree and to a vertex of color $2j$ in another good level
$j$ tree (see Figure~\ref{fig:trees} for an illustration). Each
good level $j$ tree is connected to at most one new vertex, and if
it is connected to a new vertex it stops being a good level $j$
tree.  Thus, the graph remains a forest. Therefore, the endpoints
of each batch $i$ base edge must be colored with two different
colors not in $\{1,2\ldots,2i-2\}$. With at most $2k$ colors used,
more than $g_i=a_i/(4k^2)$ of these base edges have the same pair
$c,c'$ of colors on their endpoints. We rename the colors larger
than $2i-2$ such that these two colors are called $2i-1$ and $2i$.

\begin{figure}[ht]
\begin{center}
\includegraphics[angle=90,width=0.65\textwidth]{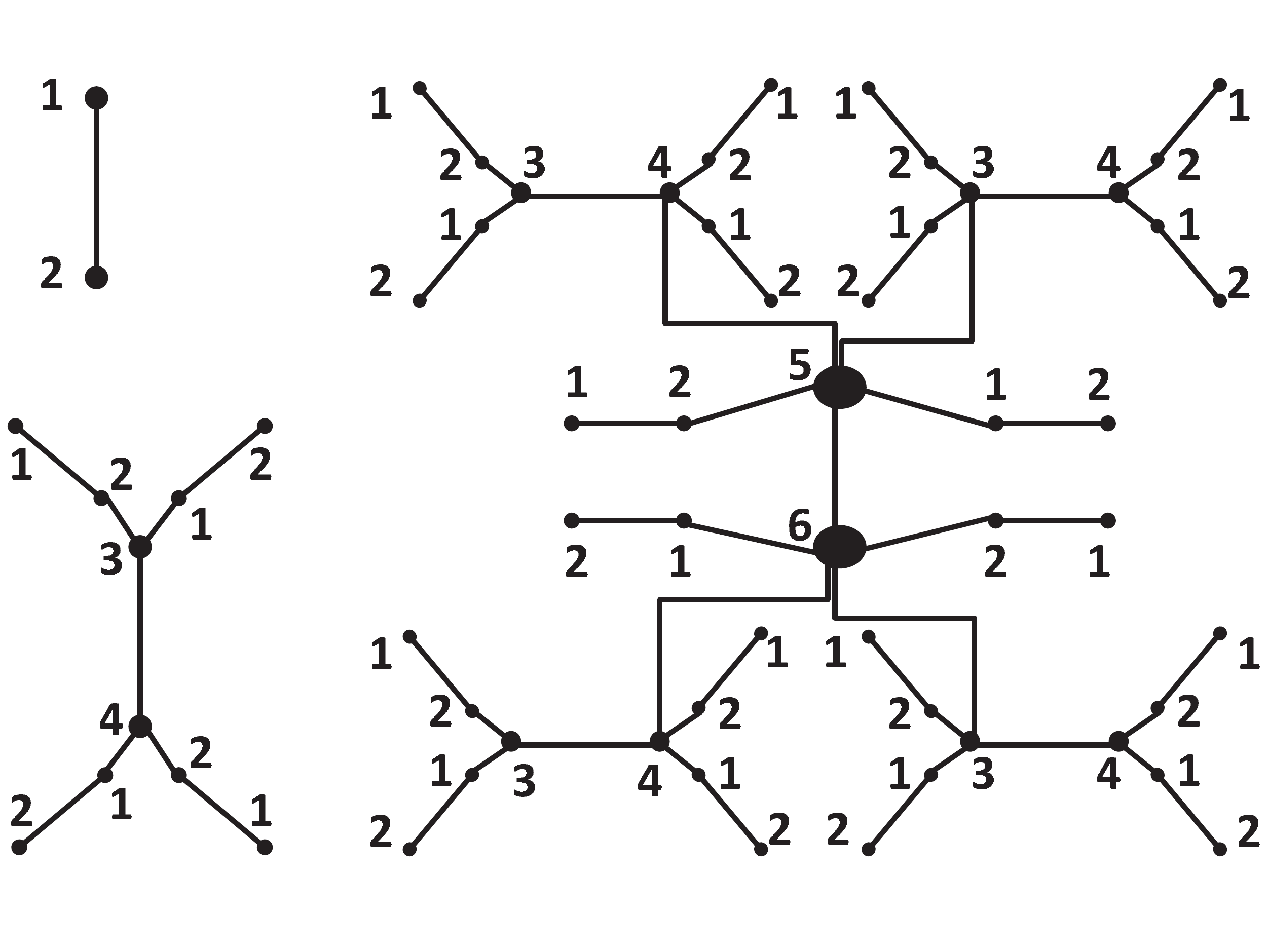}
\end{center}
\caption{The single edge on the left hand size (top) is a good
level $1$ tree. The tree on the left hand side (bottom) is a good
level $2$ tree. The tree on the right hand size is a good level
$3$ tree.} \label{fig:trees}
\end{figure}

To complete the proof, we now show that there are enough good
level $j$ trees, for $1 \leq j \leq k$. For $k=1$, this is clearly
true, since $a_k=2$ and one base edge suffices to force the
algorithm to use $2$ colors. For $k \geq 2$, we note that, for $1
\leq j<k$, good level $j$ trees are used to construct trees in all
later batches. In each batch $i>j$, exactly $2a_i$ such trees are
used. Thus, the total number of good level $j$ trees needed is
\begin{align*}
\sum_{i=j+1}^k 2a_i & = \sum_{i=j+1}^k 4 \cdot (8k^3)^{k-i}
  = 4 \cdot \sum_{i=0}^{k-j-1} (8k^3)^i
  = 4 \cdot \frac{(8k^3)^{k-j}-1}{8k^3-1}
  < \frac{(8k^3)^{k-j}}{k^3}\\
& = \frac{a_j}{2k^3}\\
& \leq \frac{a_j}{4k^2}, \text{ since } k \geq 2\\
& = g_j,
\end{align*}
where $g_j$ is the lower bound that we calculated on the number of
good level $j$ trees.

If a connected graph is required, one vertex can be added to the
last batch, connecting all trees remaining in the forest.
\mbox{}\qed\end{proof}

It is easy to see that the construction in Lemma~\ref{lbtreep} can
be changed to use many fewer vertices when considering bipartite
graphs, rather than restricting to trees. For each bipartite graph
in level $i$, it is sufficient to attach four good bipartite
graphs from level $i-1$.

The following lemma holds for any graph, not only trees.
\begin{lemma}
\label{ubbip}
There is a strictly $k$-competitive algorithm for $k$-batch coloring,
even if $k$ is not known in advance. 
\end{lemma}
\begin{proof}
Consider a graph $G$ presented in $k$ batches and let $\chi = \chi(G)$.
Since each batch of vertices induces a subgraph of $G$, each batch can
clearly be colored with at most $\chi$ colors.
Thus, using the colors $(i-1)\chi+1, (i-1)\chi+2, \ldots ,
(i-1)\chi+\chi$ for batch $i$ yields a strictly $k$-competitive algorithm.
\mbox{}\qed\end{proof}

Theorem~\ref{thm:bipartite} below follows
directly from Lemmas~\ref{lbtreep} and~\ref{ubbip}.
\begin{theorem}
\label{thm:bipartite}
For any graph class containing trees as a special case, 
the optimal (strict) competitive ratio for $k$-batch coloring is
$k$, regardless of whether or not $k$ is known in advance.
\end{theorem}

\section{Interval Coloring in Two Batches}
Since not all trees are interval graphs, the lower bound from
the previous section does not apply here.
For the case of interval graphs we show the surprising result
that while coloring in two batches has a tight bound of~$2$,
the problem becomes easier if we assume that the vertices of
the graph are revealed together with their interval representation
(and this interval representation of vertices of the first batch
cannot be modified in the second batch).  
The standard results for online coloring of interval
graphs do not make this distinction: The lower bound is obtained
for the (a priori easier) case where the interval representation
of a vertex is revealed to the algorithm when the vertex is
revealed, while the upper bound holds even if such a
representation is not revealed 
(the online
algorithm only computes a maximum clique size containing the new
vertex and applies the First-Fit algorithm on a subset of the
vertices).  Throughout this section, our lower bounds are with
respect to the asymptotic competitive ratio while our upper bounds
are for the strict competitive ratio, and thus the results are
tight for both measures.

\subsection{Unknown interval representation}
We start with a study of the case where the algorithm is guarantied
that the resulting graph (at the end of every batch) will be an
interval graph, but the interval representation of the vertices of the
first batch is not revealed to the algorithm (and may depend on the
actions of the algorithm).  We show that in this case $2$ is the
best competitive ratio that can be achieved by an online algorithm.

\begin{theorem}\label{no-rep}
For the problem of $2$-batch coloring of interval graphs with unknown interval
representation, the optimal (strict) competitive ratio is $2$.
\end{theorem}
\begin{proof}
The upper bound follows from Lemma~\ref{ubbip}. Each of the two induced
subgraphs is an interval graph, and it can be colored optimally in
polynomial time even if the interval representation is not given.

Next, we show a matching lower bound.
For a given $q \in \nats$, let $N_1 = \binom{4q}{q}+1$ and
$N_2 = \binom{4q}{2q}+1$.
In the first batch, the adversary gives $N_1+N_2$ pairwise
nonoverlapping cliques: $N_1$ cliques of size $q$ and $N_2$ cliques of
size $2q$.

Assume that an algorithm uses at most $4q$ colors for the first
batch.
By the pigeon hole principle, there are two cliques of size $q$ that
are colored with the same set $\colorset_1$ of colors.
The vertices of these two cliques will correspond to the intervals
$[5,6]$ and $[9,10]$, respectively.
Similarly, there are two cliques of size $2q$ that are colored with
the same set $\colorset_2$ of colors.
For one of these cliques, $q$ vertices will correspond to the interval
$[0,1]$ and the remaining $q$ vertices will correspond to the interval
$[0,3]$.
If any of these $2q$ vertices are colored with colors from $\colorset_1$,
they will correspond to the interval $[0,1]$.
We let $\colorset_2'$ denote the set of colors used on the vertices
corresponding to the interval $[0,3]$.
Note that $\colorset_1 \cap \colorset_2' = \emptyset$, and hence,
$|\colorset_1 \cup \colorset_2'|=2q$.
For the other of these two cliques, the $q$ vertices colored with
$\colorset_2'$  will correspond to
the interval $[12,15]$ and the remaining $q$ vertices will correspond
to the interval $[14,15]$.
All other intervals are placed to the right of the point $15$ so that
they do not overlap with any of the four cliques just described.

The second batch consists of $q$ vertices corresponding to the
interval $[2,8]$ and $q$ vertices corresponding to the interval
$[7,13]$.
All of these $2q$ intervals overlap with each other and with intervals
of all colors in $\colorset_1 \cup \colorset_2'$.
Thus, the algorithm uses at least $4q$ colors.

Since no clique is larger than $2q$ vertices, \opt uses $2q$ colors.
Since $q$ can be arbitrarily large, this proves that no deterministic
online algorithm can be better than $2$-competitive, even when
considering the asymptotic competitive ratio.
\mbox{}\qed\end{proof}

\subsection{Known interval representation}
We now assume that the vertices are revealed to the
algorithm together with their interval representation.  For this
case, we show an improved competitive ratio of $\frac 32$. We
first state the following lower bound whose proof is a special case
of the lower bound proof of Kierstead and Trotter~\cite{KT81}.

\begin{lemma}
\label{lowerTBknown}
For the problem of $2$-batch coloring of interval graphs with known interval
representation, no algorithm can achieve a competitive ratio
strictly smaller than $\frac 32$.
\end{lemma}
\begin{proof}
This is a special case of the lower bound of Kierstead and
Trotter~\cite{KT81}. The construction is as follows. For a given
$q \in \nats$, let $N = \binom{3q}{2q}$. In the first batch, the
adversary gives $2q$ vertices corresponding to the interval
$[4i,4i+1]$, for $i=0,1,\ldots,N$. Thus, the first batch consists
of $N+1$ pairwise nonoverlapping cliques. The clique of intervals
$[4i,4i+1]$ is called clique $i$.

Assume that an algorithm colors the first batch using at most $3q$
colors. Since there are more than $N$ cliques, there must be two
cliques, clique $k$ and clique $\ell$, colored with the same $2q$
colors. Assume that the cliques are named such that $k<\ell$.

In the second batch, the adversary gives $2q$ vertices
corresponding to each of the intervals $[4k,4k+3]$ and
$[4k+2,4\ell+1]$. To color the second batch vertices, the
algorithm will need $4q$ colors different from the $2q$ colors
used on clique $k$ and clique $\ell$. Thus, the algorithm uses at
least $6q$ colors in total.

Since there is no clique larger than $4q$ vertices, \opt uses only
$4q$ colors. Since $q$ can be chosen arbitrarily large, this shows
that no deterministic online algorithm can be better than
$\frac32$-competitive.
\mbox{}\qed\end{proof}

For the matching upper bound,
we define an algorithm, \algTwoBatches, which is
strictly $\frac32$-competitive,
using Algorithm~\ref{alg-first} to color the first
batch of intervals and Algorithm~\ref{alg-second} to color the
second batch. Intervals can be open, closed, or semi-closed.

Let $\omega$ denote the maximum clique size in the full graph
consisting of intervals from both batches.
For any interval $I$, let \icolor{I} denote the color assigned to
$I$ by \algTwoBatches.
Similarly, for a set \setofintervals of intervals,
\icolor{\setofintervals} denotes the set of colors used to color the
intervals in \setofintervals.

Each
endpoint of a first batch
interval $I$ is called an {\em event point}, and this event point is
associated with $I$.  If there is a point that is an endpoint of
several intervals, we have multiple copies of this
point as event points each of which is associated with a different
interval.

We define a total order, \totalorder, on the event points.
 For two event
points $p$ and $p'$, if $p<p'$, then $p$ appears before $p'$ in the
total order.
For the case $p=p'$, let $I$ and $I'$ be the
two intervals such that $p$ and $p'$ are associated with $I$ and $I'$,
respectively.
We consider a total order satisfying the following properties.
\begin{enumerate}
\item If $p$ and $p'$ are both right endpoints, $p\notin I$, and $p'\in I'$,
  then $p$   appears before $p'$ in the total order of the event points.
\item Otherwise, if $p$ is a left endpoint and $p'$ is an right endpoint, then:
\begin{itemize}
\item If $p\notin I$, then $p'$ appears before $p$ in the total order
  of the event points.
\item If $p\in I$ and $p'\notin I'$, then $p'$ appears before $p$ in
  the total order of the event points.
\item If $p\in I$ and $p'\in I'$, then $p$ appears before $p'$ in the
  total order of the event points.
\end{itemize}
\item Otherwise (that is, $p$ and $p'$ are both left endpoints),
  then if $p\in I$ and $p'\notin I'$, then $p$
  appears before $p'$ in the total order of the event points.
\end{enumerate}
We fix one particular total order, \totalorder, on the event points
satisfying all these conditions.
Observe that this total order is a
refinement of the standard order ``$\leq$'' on the real numbers.
If an event point $p$ appears before another event point $q$ in \totalorder,
we write $p \smaller q$.
If $p$, $q$, and $p'$ are event points, we say that $q$ is {\em
  between} $p$ and $p'$, if and only if $p \smaller q \smaller p'$.
In this case, we will also sometimes say that $q$ is to the right of
$p$ and to the left of $p'$.

For any three points $p$, $q$, and $p'$, where $q$ is {\em not} an
event point, we say that $q$ is between $p$ and $p'$, if and only
if $p < q < p'$. In this case, we will also sometimes say that $q$
is to the right of $p$ and to the left of $p'$.

\subsubsection{First batch.}
Algorithm~\ref{alg-first} processes the event points in the order
given by $T$. When a right endpoint is processed, we say that the
color of the associated interval is {\em released}. It is then
\emph{available} until it is used again. When processing a left
endpoint, the associated interval is colored with the most
recently released available color (or a new color, if necessary).
Thus, a stack ordering is used for the colors. Pseudo-code for the
algorithm is given in Algorithm~\ref{alg-first}.

\begin{algorithm}[ht!]
\algsetup{indent=2em}
\begin{algorithmic}[1]
\STATE $\omega_1 \gets$ maximum clique size in the first batch
\STATE Initialize an empty stack of colors
          \COMMENT{\emph{The stack will contain available colors}}
\FOR{$i \gets \omega_1$ downto $1$}
  \STATE $\push{i}$
\ENDFOR
\FOR{each event point, $p$, in the order given by $T$}
  \IF{$p$ is a {\em left} endpoint of an interval $I$}
    \STATE $\icolor{I} \gets \pop{}$
  \ELSIF{$p$ is a {\em right} endpoint of an interval $I$}
    \STATE $\push{\icolor{I}}$
  \ENDIF
\ENDFOR
\end{algorithmic}
\caption{Coloring the first-batch intervals.}
\label{alg-first}
\end{algorithm}

For ease of presentation, we insert $2\omega$ dummy intervals into
the first batch: one clique of size $\omega$ {\em before} all input
intervals and one clique of size $\omega$ {\em after} all input
intervals.
Since these dummy cliques do not overlap with any other intervals,
each will be colored with the colors $1,2,\ldots,\omega$, and they
will not influence the behavior of Algorithm~\ref{alg-first} on the
rest of the first-batch intervals.

We note that it is well-known that one can color an interval graph
with a maximum clique size of $\omega$ using $\omega$ colors, by
maintaining a set of available colors, and traversing the event
points according to the total order \totalorder: Each time a left
endpoint is considered, we color its associated interval with one
of the colors in the set of available colors (removing it from
this set), and each time a right endpoint is considered we return
the color of its interval to the set of available colors.  By
using Algorithm~\ref{alg-first}, we consider one particular
tie-breaking rule used whenever the set of available colors
contains more than one color.  More commonly, one considers the
First-Fit rule of using the minimum color in the set of available
colors.  However, in order to establish the improved bound of
$\frac 32$ on the strict competitive ratio (or even for the
competitive ratio) of the algorithm for two batches, we need to
use a different tie-breaking rule, the one defined by using a
stack as in Algorithm~\ref{alg-first}.

Before analyzing the algorithm, we introduce some additional terminology.
{\em Maximal cliques} always refer only to first-batch intervals.
For each maximal clique, we choose a point, called a {\em clique point},
contained in all intervals of the clique.
If a clique point $p$ appears to the right of another clique point
$q$, we say that the clique corresponding to $p$ appears to the right
of the clique corresponding to $q$, and vice versa.

For each maximal clique, \setofintervals, we order the intervals
of the clique by left and right endpoints, respectively, resulting
in two orderings, \leftordering{\setofintervals}{\cdot} and
\rightordering{\setofintervals}{\cdot}. The further an endpoint is
from the clique point of \setofintervals, the earlier the interval
appears in the ordering. More precisely, for each interval $I \in
\setofintervals$, $\leftordering{\setofintervals}{I} = i$, if the
left endpoint of $I$ appears as the $i$th in \totalorder among the
endpoints associated with intervals in \setofintervals. Similarly,
$\rightordering{\setofintervals}{I}=j$, if the right endpoint of
$I$ appears as the $j$th last in \totalorder among the endpoints
associated with intervals in \setofintervals. As an example,
consider the clique \setofintervals consisting of the three
intervals $a=[1,6]$, $b=[2,4]$, and $c=[3,5]$. For this clique, we
have $\leftordering{\setofintervals}{a} = 1$,
$\leftordering{\setofintervals}{b} = 2$,
$\leftordering{\setofintervals}{c} = 3$ and
$\rightordering{\setofintervals}{a} = 1$,
$\rightordering{\setofintervals}{b} = 3$,
$\rightordering{\setofintervals}{c} = 2$.

The following lemma is illustrated in Figure~\ref{fig:lemma}.

\begin{figure}[ht]
\begin{center}
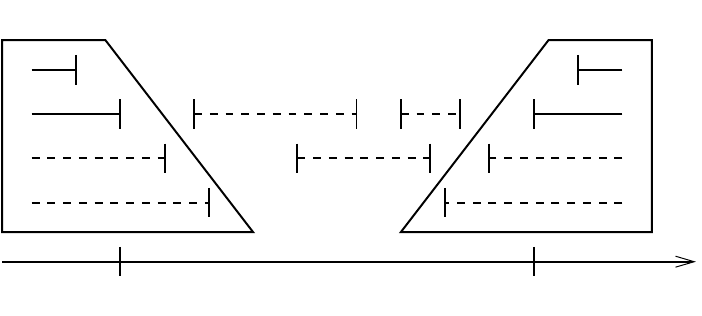
\end{center}
\caption{Illustration of Lemma~\ref{lemma:firstbatch}, with $h=3$. The intervals of
  $\setofintervals'$ are drawn with dashed lines.} \label{fig:lemma}
\end{figure}

\begin{lemma}
\label{lemma:firstbatch}
Consider a maximal clique, $\setofintervals_{\ell}$, of size $m$ and
an interval $\ileft \in \setofintervals_{\ell}$ such that
$\rightordering{\setofintervals_{\ell}}{\ileft} = \height$.
Let $\setofintervals_{\ell}^{\height} = \SETOF{I \in \setofintervals_{\ell}}{
  \rightordering{\setofintervals_{\ell}}{I}<\height}$ be the $\height-1$ intervals
in $\setofintervals_{\ell}$ with the rightmost right endpoints.
Let $\setofintervals_r$ be the first maximal clique of size at least $\height$ to
the right of $\setofintervals_{\ell}$ and let $I_r \in
\setofintervals_r$ be such that
$\leftordering{\setofintervals_r}{I_{r}} = \height$.
Finally, let $p_{\ell}$ be the right endpoint of \ileft, let $p_r$ be the
  left endpoint of \iright,
and consider the set $\setofintervals'$ of first-batch intervals
containing a point $p$ with $p_{\ell} < p < p_r$ or an endpoint $p$
with $p_{\ell} \smaller p \smaller p_r$.
Then, $\icolor{\setofintervals'}
\subseteq \icolor{\setofintervals_{\ell}^{\height}} \,.$
\end{lemma}
\begin{proof}
For the intervals in $\setofintervals_{\ell}^{\height}$, the lemma follows
trivially.

We now consider the intervals in $\setofintervals' \setminus
\setofintervals_{\ell}^{\height}$. Note that we can assume $p_{\ell}
\smaller p_r$, since otherwise $\setofintervals'$ is empty, and
the lemma follows trivially.

Assume for the sake of contradiction that some interval $I' \in
\setofintervals' \setminus \setofintervals_{\ell}^{\height}$ has an
endpoint to the left of $p_{\ell}$. This interval would overlap
with all intervals in $\setofintervals_{\ell}^{\height}$, contradicting
the assumption that $p_{\ell} \smaller p_r$. Hence, we only need
to consider intervals with a left endpoint between $p_{\ell}$ and
$p_r$.

Consider any interval with a left endpoint $p$ such that $p_{\ell}
\smaller p \smaller p_r$. It follows from the definition of
$\setofintervals_{r}$ and $I_r$ that there are more right
endpoints than left endpoints between $p_{\ell}$ and $p$ (note
that $\height>1$ in this case). Thus, at $p$, the most recently released
available color is a color in \icolor{\setofintervals_{\ell}^{\height}},
and therefore, the interval associated with $p$ is colored with a
color in \icolor{\setofintervals_{\ell}^{\height}}.
\mbox{}\qed\end{proof}

\subsubsection{Second batch.}
We now describe the algorithm, Algorithm~\ref{alg-second}, given
in pseudo-code below, for coloring the second batch intervals.

A {\em chain} is a set of nonoverlapping second batch intervals.
The algorithm starts with partitioning the second-batch intervals into
$\omega$ chains (some of which may be empty).
This is clearly possible, since the graph is $\omega$-colorable.

The second batch intervals are colored in iterations, two chains per
iteration.
The algorithm keeps a counter, $i$, which is incremented once in each
iteration, and maintains the set \coloredBatchTwo of second batch
intervals that the algorithm has already colored.
In each iteration, a set of nonoverlapping first-batch intervals is
\emph{processed}.
The algorithm maintains the invariant that, at the beginning of each
iteration, any maximal first-batch clique of size $\height$ contains exactly
$\min\{\height,\omega-i\}$ unprocessed intervals.

A first-batch maximal clique of size at least $\omega-i+1$ as well
as its clique point is said to be {\em active}.  The part of the
real line between two neighboring active clique points is called a
{\em region}. Throughout the execution of
Algorithm~\ref{alg-second}, the number of regions is
nondecreasing, and whenever a region is split, the chains of the
region are also split by a simple projection onto each region and
each resulting region will contain its boundary active clique
points (in particular, this means that active clique points may
belong to two regions). In each iteration, each region and its
chains are treated separately.

The algorithm maintains the invariant that no uncolored second batch
interval overlaps with more than one region.  This is the key property,
allowing the algorithm to consider one region at a time in a given
iteration of the algorithm.  First-batch intervals overlapping with
more than one region will be cut into more intervals, with a cutting
point at each active clique point contained in the interval.  Thus, by
cutting the intervals of an active clique of size $\height$, the clique is
replaced by two cliques of size $\height$ in neighboring regions.
When a first-batch interval is cut into parts, the different
parts of the interval may be processed in different iterations, but no
new event points are introduced.

In the $i$th iteration, one chain in each region is colored with the
color of a first-batch interval in the region being processed in this
iteration, and another chain of the region will be colored with the
color $\omega+i$, which has not been used in the region before.
For any point $p$, let $\pdepth$ be the number of second batch intervals
containing $p$.
We say that $p$ is {\em covered} by a set $S$ of
intervals, if there are $\min\{\pdepth,i\}$ second batch intervals in $S$
containing $p$.

Next, we define a set $\SP$ of {\em representative points}, such that
each interval between two neighboring clique points is represented by
one point.
To this end, we define the following equivalence relation between points on the real
line.  For a pair of points $p$ and $p'$, we say that $p$ \emph{is equivalent
to} $p'$ if the following conditions hold:
\begin{itemize}
\item For every clique point $q$, either both $p\leq q$ and $p'\leq q$
or both $p \geq q$ and $p' \geq q$
  (this in particular means that $p$ and $p'$ belong to a common
  region).
\item For every interval $I$ of either the first batch or the second
  batch, we have that either $p,p' \in I$ or the two points $p$ and
  $p'$ are to a common side of $I$ (either both are smaller than any
  point in $I$ or both are larger than any point in $I$).
\end{itemize}
Observe that the number of equivalence classes of this relation is
linear in the number of intervals of the input.  The set of representative
points $\SP$ is defined such that each equivalence class has
exactly one member in $\SP$, chosen arbitrarily.  We use the
following observation.

\begin{observation}\label{obs_cover}
  For a given set $S$ of intervals, at any point after
  line~\ref{initi} in Algorithm~\ref{alg-second}, we have that all
  points (on the real line) are covered by $S$ if and only if all
  points in $\SP$ are covered by $S$.
\end{observation}

For a region $R$, we denote by $\rSP$ the set of representative
points contained in region $R$ (that is, $\rSP = R \cap \SP$).

\begin{algorithm}[h!]
\algsetup{indent=0.9em}
\begin{algorithmic}[1]
\STATE Mark all first-batch intervals as unprocessed 
\STATE Create
an optimal coloring of the second-batch intervals, using a set \colorset of
 $\omega$ colors \STATE $R \gets (-\infty,\infty)$
\COMMENT{{\em Initially, there is only one region}} 
\STATE
$\rchains \gets \emptyset$ \STATE $\rSP \gets$ the set of
representative points in region $R$
\FOR{each color $c \in
\colorset$}
  \STATE $\rchains \gets \rchains \cup \{\{I \mid I \text{
    is a second batch interval with color } c\}\}$ \label{initrchain}
\ENDFOR
\STATE $\coloredBatchTwo \gets \emptyset$ \COMMENT{{\em Set of colored
  second batch intervals}}
\STATE $i \gets 0$ \label{initi}
\WHILE[{\em Invariant $I$}]{$i < \floor{\omega/2}$}
  \STATE\COMMENT{\textcolor{blue}{\em Color two chains:}}
  \STATE $i \gets i+1$ \label{inci}
  \STATE Split all regions (incl.\ the assoc.\ chains and sets
  of repr.\ points) at all active clique points
         \label{splitChains}
  \FOR{each region $R$ containing at least one nonempty chain}
    \STATE $(\chainOne,\chainTwo) \gets$ {\sc CreateChains}($R$) \COMMENT{See Algorithm~\ref{alg-chains}}
    \STATE \COMMENT{\textcolor{blue}{\em Color intervals in
        \chainOne and \chainTwo using a first batch color and a new color:}}
    \STATE $I_{\ell} \gets$ the unprocessed first-batch interval
    of the earliest event point in $R$
    \STATE $I_r \gets$ the unprocessed first-batch interval
    of the latest event point in $R$
    \STATE Mark $I_{\ell}$ and $I_r$ as processed
    \STATE Give all intervals in \chainOne the color of $I_{\ell}$
           \label{colorChainOne}
    \STATE Give all intervals in \chainTwo the color $\omega + i$
           \label{colorChainTwo}
    \STATE $\coloredBatchTwo \gets \coloredBatchTwo \cup \chainOne
    \cup \chainTwo$\label{b2colored}
    \STATE $\rchain \gets \rchain \setminus \{ \chainOne, \chainTwo \}$
  \ENDFOR
\ENDWHILE
\STATE \COMMENT{\textcolor{blue}{\em If $\omega$ is odd, each region
    may have one chain left to color:}}
\FOR{each region $R$ where \rchains contains a nonempty chain \chain}
  \STATE $I \gets$ the unprocessed first-batch interval with the
         earliest event point in $R$
  \STATE Give the intervals of \chain the color of $I$
         \label{colorLastChain}
\ENDFOR
\end{algorithmic}
\caption{Coloring the second batch intervals.}
\label{alg-second}
\end{algorithm}

\begin{algorithm}[h!]
\algsetup{indent=0.9em}
\begin{algorithmic}[1]
    \STATE \chainOne $\gets$ a chain in \rchains containing
           the leftmost left endpoint
         \STATE \chainTwo $\gets$  any other chain from \rchains
    \WHILE{some point in $\rSP$ is {\em not} covered by $\coloredBatchTwo \cup
      \chainOne \cup \chainTwo$ \label{whileStart}}
      \STATE $p \gets$ the leftmost point in $\rSP$ {\em not} covered by $\coloredBatchTwo \cup
      \chainOne \cup \chainTwo$
      \STATE \chainThree $\gets$ a chain from \rchains containing $p$ \label{chainThree}
      \IF{for all points $q < p$ in $\rSP$,
               $q$ is contained in \chainThree \OR
               in both \chainOne and \chainTwo}
        \STATE \chainTwo $\gets$ \chainThree{}
             \COMMENT{\emph{\chainTwo now refers to the chain in \rchains that \chainThree refers to}}
      \ELSE
        \STATE $q \gets$ the rightmost point in $\rSP$ left of $p$ violating the
               condition
        \STATE $\chain \gets$ one of \chainOne or \chainTwo not containing $q$
               \COMMENT{\emph{\chain now refers to a chain in \rchains}}
        \STATE \COMMENT{\textcolor{blue}{\em Do a crossover of \chain and \chainThree
            at the point $q$, modifying \chain and \chainThree in \rchains:}}
        \STATE $\textsc{Tail} \gets \{ I \in \chain \mid I \text{
               starts to the right of } q \}$
               \STATE $\textsc{Tail}_3 \gets \{ I \in \chainThree \mid I \text{
               starts to the right of } q \}$
        \STATE $\chain \gets (\chain \setminus \textsc{Tail})
               \cup \textsc{Tail}_3$
        \STATE $\chainThree \gets (\chainThree \setminus \textsc{Tail}_3)
               \cup \textsc{Tail}$
      \ENDIF
    \ENDWHILE \label{whileEnd}
    \RETURN{$(\chainOne,\chainTwo)$}
\end{algorithmic}
\label{alg-chains}
\caption{{\sc CreateChains}($R$)}
\end{algorithm}

We use the following loop invariant for each region to establish
that  \algTwoBatches is correct and strictly $3/2$-competitive. The
proof of the invariant $I$ is based on induction on the value of
$i$.

\medskip\noindent
{\bf Invariant $\boldsymbol{I}$:}
\begin{enumerate}
\item[(I1)] All points $p$ are covered by the set \coloredBatchTwo.
\item[(I2)] No color used for an unprocessed first-batch interval contained in
a region $R$
  has been used for a second batch interval intersecting region $R$ so far.
\item[(I3)] Each active clique has exactly $\omega-i$ unprocessed
  intervals.
\item[(I4)] For each region $R$, \rchain has at most $\omega-2i$ chains.
\end{enumerate}

\begin{lemma}
$I$ is an invariant for the \textbf{while}-loop in Algorithm~\ref{alg-second}.
\end{lemma}
\begin{proof}
We prove by on induction on $i$ that the invariant
holds at the start of every iteration of the \textbf{while}-loop.
\begin{enumerate}
\item[(I1)]
  By Observation~\ref{obs_cover}, it suffices to show that
  the set \coloredBatchTwo covers all points
  in $\SP$. 

  At the beginning of the first iteration, $i=0$, so (I1) is trivially
  true.
  At the beginning of each of the following iterations, it follows
  from~(I1) that each point $p$ is contained in at least $\min\{\pdepth,i\}$
  intervals in \coloredBatchTwo.
  In line~\ref{b2colored}, all intervals of $\chainOne \cup \chainTwo$
  are added to \coloredBatchTwo.  Thus, we only need to prove that, if
  $p$ is not covered after the increment of $i$ in line~\ref{inci}, the
  while loop in lines~\ref{whileStart}--\ref{whileEnd}
  of Algorithm~\ref{alg-chains}
  will add at least one interval containing
  $p$ to $\chainOne \cup \chainTwo$.

  Consider a region $R$ and let $p_{\ell}$ and $p_r$ be defined as in
  Lemma~\ref{lemma:firstbatch}, with $\height = \omega-i+1$.
  Any point in $R$ to the left of $p_{\ell}$ or to the right of $p_r$
  is contained in at least $\height$ first-batch intervals.
  Hence, there cannot be any uncovered points in $R$ to the left of
  $p_{\ell}$ or to the right of $p_r$.

  As long as some point $p$ in $\rSP$ is not covered by $\coloredBatchTwo \cup
  \chainOne \cup \chainTwo$, it follows from~(I1) and
  the definition of covered that $\rchains$ contains $p$ and that
  $\chainOne \cup \chainTwo$ does not contain $p$.
  Thus, \chainThree of line~\ref{chainThree} of Algorithm~\ref{alg-chains} exists.

  If for all points $q<p$ in $\rSP$,
  \chainThree contains $q$ or both \chainOne and
  \chainTwo contain $q$, swapping any of
  \chainOne or \chainTwo with \chainThree will ensure that $\chainOne
  \cup \chainTwo$ contains all points $q \leq p$ in $\rSP$.

  Otherwise, there is a point $q < p$ in $\rSP$ not contained in \chainThree
  and not contained in both \chainOne and \chainTwo.
  The algorithm chooses $q$ as the rightmost such point
  among the points in $\rSP$.
  The algorithm then chooses a $\chain \in \{\chainOne,\chainTwo\}$ such
  that \chain does not contain $q$.
  Since neither \chain nor \chainThree contains $q$, all intervals in
  $\chain \cup \chainThree$ appear strictly before or after $q$.
  Thus it is possible to
  cut each of \chain and \chainThree into two sets, ``head'' and
  ``tail'' consisting of the intervals ending before $q$ or starting after
  $q$, respectively, and let the two chains swap tails, while maintaining the
  property that no two intervals within a chain overlap.

  After this crossover operation, $p$ is contained in \chain.
  Neither $\chainOne$ nor $\chainTwo$ is changed to the left of $q$.
  Since all points between $q$ and $p$ were contained in \chainThree
  or both \chainOne and \chainTwo before the crossover, all such points
  are still contained in $\chainOne \cup \chainTwo$ after the crossover.
  Thus, the leftmost point in $\rSP$ not covered by $\coloredBatchTwo \cup
  \chainOne \cup \chainTwo$ now occurs at or to the right of the leftmost
  point in $\rSP$ to the right of $p$.
  This means that, after $O(n)$
  iterations of the while loop of
  lines~\ref{whileStart}--\ref{whileEnd} of Algorithm~\ref{alg-chains},
  all points in $\SP$ are
  covered by $\coloredBatchTwo \cup \chainOne \cup \chainTwo$.
\item[(I2)] At the beginning, the statement is trivially true,
since no
  second-batch interval has been colored.

  Since no unprocessed first-batch intervals in a region
  are colored with $\icolor{\ileft}$, according to
  Lemma~\ref{lemma:firstbatch}, and since no
  first-batch intervals are colored with $\omega+i$, (I2) is
  maintained.
\item[(I3)] At the beginning of the first iteration, (I3) is
trivially
  true, since every active clique has $\omega$ first-batch intervals,
  and all first-batch intervals are unprocessed.

  In each iteration, the cliques of size $\omega-i+1$ are added to the
  set of active cliques, and one interval of each active clique is
  processed.
  Hence, (I3) is maintained.

\item[(I4)] At the beginning, the statement holds as an optimal
coloring
  consists of exactly $\omega$ color classes and the number of chains
  in $\rchain$ after line~\ref{initrchain} is the number of color classes.

  In each iteration of the while loop in
  lines~\ref{whileStart}--\ref{whileEnd} of Algorithm~\ref{alg-chains},
  the number of chains in
  \rchain is not modified, as in each such iteration we replace two
  chains by another pair of chains covering the same set of second
  batch intervals.  Invariant~(I4) is maintained because the number of
  chains in \rchain is modified once in every iteration in
  line~\ref{b2colored}, where it is decreased by two.
\end{enumerate}
\mbox{}\qed\end{proof}

We use the invariant $I$ to prove that
for any input $\sigma$, \algTwoBatches produces a proper
coloring using at most $\FLOOR{\frac32 \opt(\sigma)}$ colors.

\begin{lemma}
\label{lemma:algTB}
For any input $\sigma$, the algorithm \algTwoBatches produces a proper
coloring using at most $\FLOOR{\frac32 \opt(\sigma)}$ colors.
\end{lemma}

\begin{proof}
We first note that, by~(I1), no chain in \rchains can contain an
active clique point. Thus, the splitting of chains done in
line~\ref{splitChains} is possible.

In each of the $\floor{\omega/2}$ iterations of the while loop of
lines~\ref{whileStart}--\ref{whileEnd} of Algorithm~\ref{alg-chains},
two chains are colored and
the number of chains in \rchain is decreased by two. If $\omega$
is odd, one additional chain may be colored in
line~\ref{colorLastChain} (and at this point \rchain consists of a
single chain by invariant~(I4)). Thus, all of the $\omega$ chains
containing all second batch intervals are colored.

The actual coloring happens in lines~\ref{colorChainOne}
and~\ref{colorChainTwo}, and possibly in
line~\ref{colorLastChain}.
In line~\ref{colorChainOne}, the color used is the color, $c$, of
the earliest event point associated with the unprocessed
first-batch interval, $I_{\ell}$, in the region $R$.
By Lemma~\ref{lemma:firstbatch} and (I3), \ileft and \iright are
the only first-batch intervals overlapping with $R$ that are
colored with $c$.
By invariant (I1) and the definition of \ileft and \iright, no
interval in \rchains overlaps with \ileft and \iright in $R$.
By invariant (I2), no second batch interval in $R$ has been
colored with $c$ in earlier iterations.
Thus, coloring the intervals of \chainOne results in a legal
coloring of these intervals. The same arguments hold for the
possible coloring done in line~\ref{colorLastChain}.
Since the color $\omega + i$ has never been used before, the
coloring of \chainTwo is also legal.

In summary, in Algorithm~\ref{alg-first}, $\omega$ colors are
used. In Algorithm~\ref{alg-second}, the colors are used again for
some intervals and only $\FLOOR{\omega/2}$ new colors are used.
\mbox{}\qed\end{proof}

Combining Lemmas~\ref{lowerTBknown} and~\ref{lemma:algTB}
shows that the optimal (strict) competitive
ratio for the problem is $\frac32$:
\begin{theorem}
\label{thm:algTB}
\algTwoBatches has a strict competitive ratio of $\frac32$.
\end{theorem}

\section{Sum Coloring of Graphs in Multiple Batches}
We study two cases separately: the case where the number of
batches is known to the algorithm from the beginning, and the case where
it is not.  Once again, our lower bounds are for the competitive
ratio and our upper bounds are for the strict competitive
ratio.

\subsection{Number of batches known in advance}
We start our study of sum coloring by examining the case where the
algorithm knows the number of batches $k$ in advance. Recall that
we do not require that algorithms used within one batch be
polynomial time.

\begin{lemma}\label{knownk-ub}
There is a strictly $k$-competitive algorithm for sum coloring in
$k$ batches, if $k$ is known in advance.
\end{lemma}
\begin{proof}
For each batch, the algorithm, \kbatch, applies an optimal
procedure, \optcolor, to compute an optimal sum coloring for the
subgraph induced by the set of vertices of batch $i$, separately
from previous batches. In order to construct the solution of the
input graph, \kbatch applies the following transformation: For
every vertex $v$ of batch $i$, if \optcolor colors $v$ with color
$c$, then \kbatch colors $v$ using color $f(i,c)=k\cdot (c-1)+i$.
This function $f$ satisfies $f(i,c) \equiv i \ (\bmod \  k)$, so
if $f(i,c)=f(i',c')$, for some $1 \leq i,i' \leq k$, then $i=i'$.
Moreover, if $f(i,c)=f(i,c')$, then $k(c-c')=0$, and therefore
$c=c'$. Thus, vertices of different batches have different colors,
and two vertices of the same batch have the same color after the
transformation if and only if they had the same color in the
solution returned by \optcolor. As any proper coloring of the graph
provides proper colorings for the $k$ induced subgraphs, the total
cost of the $k$ outputs of \optcolor does not exceed the cost of
an optimal coloring of the
entire graph. 
For any color $c$ and batch $i$, $f(i,c) \leq
k\cdot c$. Thus, the cost of the output is at most $k$ times the
total cost of the $k$ solutions returned by \optcolor (for the
$k$ vertex disjoint induced subgraphs).
\mbox{}\qed\end{proof}

We prove a matching lower bound for this case, which holds even for the
asymptotic competitive ratio.
\begin{lemma}
\label{knownk-lb}
No algorithm for sum coloring in $k$ batches has a competitive ratio
strictly smaller than $k$, even if $k$ is known in advance.
\end{lemma}

\begin{proof}
Assume for the sake of contradiction that there is a value of
$k$ and an online algorithm \alg for sum coloring of graphs in $k$
batches whose competitive ratio $\rho$ is strictly smaller than
$k$.  Let $M$ be a large integer such that $M >
\max\{2k^2,\frac{2\rho}{k-\rho}\}$.

The algorithm will be presented with $k$ batches, such that after
every batch $i$ ($1 \leq i \leq k$) the input either stops  (the
remaining $k-i$ batches will be empty), or one vertex of the
batch, which will be denoted by $v_i$, is selected as a designated
vertex, and it will be used for constructing the other batches.

Batch $i$ (for $i=1,2,\ldots,k$) is constructed as follows: The
batch consists of a set $V_i$ of $M^i$ vertices, each of which has
$i-1$ neighbors that are $v_1,v_2,\ldots ,v_{i-1}$ (thus, the
vertices of $V_i$ form an independent set and the vertices $v_1,
v_2, \ldots, v_i$ form a clique). For $i\leq k-1$, if the
algorithm colors all vertices of batch $i$ with colors of value at
least $k$, then the input stops. Otherwise, one vertex whose color
is in $\{1,2,\ldots,k-1\}$ is selected to be $v_i$, and the set of
vertices of the next batch, $V_{i+1}$, is presented. If the input
consists of $j$ batches, $V=\cup_{i=1}^j V_i$.

We compute an upper bound on the optimal sum of colors, if the
input stops after the first $j$ batches (we describe solutions which
are not necessarily optimal).  If the input consists of $j$
batches, we next show that the set
$V\setminus\{v_1,v_2,\ldots,v_{j-1}\}$ is independent. Consider a
vertex $v$ of batch $i$. This vertex is presented with edges to
$\{v_1,v_2,\ldots,v_{i-1}\}$. If $v$ does not become $v_i$, it
will not have any further edges. Thus, it is possible to assign
color $1$ to each such vertex, and use color $i+1$ for $v_i$. We
show that, for $1 \leq j \leq k$, this gives a total cost of $O_j
\leq M^j+2M^{j-1}$. For $j=1$ and $j=2$, we obtain $O_1=M$ and
$O_2=M^2+M+1\leq M^2+2M$. For $j \geq 3$, the total cost of this
solution is
\begin{align*}
O_j &=\sum_{i=1}^j (M^i-1)+1+\sum_{i=1}^{j-1} (i+1)=\sum_{i=1}^j
M^i+\sum_{i=1}^{j-1}
i=\frac{M^{j+1}-1}{M-1}-1+\frac{j(j-1)}{2}\\ 
&<\frac{M^{j+1}}{M-1}+k^2<\frac{M^{j+1}}{M-1}+M^{j-2}, \text{ as
$j-2\geq 1$.}
\end{align*}
We find that
\begin{alignat*}{2}
&& \frac{M^{j+1}}{M-1}+M^{j-2} & \leq M^j + 2M^{j-1} \\
\Leftrightarrow && \frac{M^3}{M-1} + 1 & \leq M^2+2M\\
\Leftrightarrow && \; M^3 + M - 1 & \leq M^3 + 2M^2 - M^2 - 2M\\
\Leftrightarrow && 3M & \leq M^2+1, \text{ which holds for any } M
\geq 3.
\end{alignat*}

If the input stops after $j<k$ batches, then \alg has colored
$M^j$ vertices with colors of at least $k$, and its cost is at
least $k\cdot M^j$. Otherwise, consider batch $k$. Each of the
vertices $v_1,v_2,\ldots ,v_{k-1}$ was given a color no larger
than $k-1$, and since they induce a clique, each of the colors
$\{1,2,\ldots,k-1\}$ is used exactly once on these $k-1$ vertices.
When the set $V_k$ is presented, each vertex of $V_k$ is connected
to each vertex in $\{v_1,v_2,\ldots ,v_{k-1}\}$, so every vertex
of $V_k$ must be colored with a color that is at least $k$. Thus,
in this batch the total cost of the algorithm is at least $k\cdot
M^k$.

We showed that if the input stops after batch $j$ (for $1 \leq j
\leq k$), the cost of the algorithm is at least $k\cdot M^j$,
while the cost of an optimal solution does not exceed
$M^j+2M^{j-1}$. The performance ratio is thus at least
$$\frac{k\cdot M^j}{M^j+2M^{j-1}} =
\frac{k}{1+2/M}>\frac{k}{1+2(k-\rho)/(2\rho)}=\frac{2\rho
k}{2\rho+2k-2\rho}=\rho\,,$$ as  $M>\frac{2\rho}{k-\rho}$. This
contradicts the assumption that the competitive ratio was $\rho$.
\mbox{}\qed\end{proof}

Combining Lemmas~\ref{knownk-ub} and~\ref{knownk-lb} gives the following result:
\begin{theorem}
\label{thm:knownk}
For sum coloring in $k$ batches, with $k$ known in
advance, the optimal (strict) competitive ratio is $k$.
\end{theorem}

\begin{remark}\label{rmksum}
Observe that the graph in the proof of Lemma~\ref{knownk-lb} for the case $k=2$ 
has no cycles. Thus, there is no online algorithm for sum coloring of forests
in $k$ batches with competitive ratio strictly smaller than~$2$. This can
be strengthened to trees by adding one extra vertex in the second batch which is
adjacent to all of the isolated vertices from the first batch.
\end{remark}

\begin{theorem}\label{ff_for_sum_col}
For sum coloring of trees in $k$ batches,  First-Fit 
is strictly $2$-competitive, and this is the best possible competitive
ratio, even if $k$ is known in advance.
\end{theorem}
\begin{proof}
The lower bound of $2$ for any online algorithm holds because the
graph in the proof of Lemma~\ref{knownk-lb} for the
  case $k=2$ has no cycles, and thus there is no online algorithm for
  sum coloring of forests in $2$ batches whose competitive ratio is
  strictly smaller than 2. This can be extended to trees by adding one
  extra vertex in the second batch which is adjacent to all
  vertices of the first batch.
To prove the upper bound of First-Fit, we will show by induction
on $t$ that when First-Fit is used for coloring a tree (a
connected subgraph) on $t$ vertices the sum of the colors of the
vertices is at most $2t-1$. Obviously, no algorithm can have a
cost below $t$ (in fact, if $t>1$, then any coloring has cost at
least $t+1$).

For $t=1$, the claim follows trivially as a single vertex is
assigned color $1$.  Assume that the claim holds for all $t'<t$
and we prove it for $t$.  Consider the last vertex to be colored
by First-Fit.  This vertex will be connected to some number of
existing trees (connected components of the graph prior to this
iteration), and we denote this number (of components) by $X$. If
$X=0$, then the new vertex gets color $1$ and becomes a singleton
(so $t=1$ and the sum of colors is $1=2t-1$). If $X>0$, then for
every $j=1,2,\ldots ,X$, the $j$-th tree with $t_j$ vertices had
the sum of colors $2t_j-1$ (or less). The new vertex has a color
not exceeding $X+1$ (as the new vertex is connected to one vertex
of each of the $X$ existing trees), and we have $\sum_{j=1}^X t_j
=t-1$. We find that the total cost of the solution returned by
First-Fit does not exceed $\sum_{j=1}^X
(2t_j-1)+X+1=2(t-1)-X+(X+1)=2t-1$.
\mbox{}\qed\end{proof}

\subsection{Number of batches unknown in advance}
Next, we consider the case where the number of batches $k$ is not
known in advance.  
Thus, to obtain a given competitive ratio, this ratio must be obtained
after each batch.
Note that the algorithm described in the proof of
Lemma~\ref{knownk-ub} cannot be used in this case. While the
algorithm is not well defined if $k$ is unknown in advance to the
algorithm, it may seem that modifying the value of $k$ by doubling
would result in a competitive ratio of $O(k)$, but no such
algorithm exists. We prove that for any positive nondecreasing
sequence $f(i)$, which is defined for integer values of $i$ (where
$f(i)\geq 1$ for $i\geq 1$), no algorithm with competitive ratio
$O(f(k))$ can be given if the series $S_f=\sum_{i=1}^{\infty}
\frac{1}{f(i)}$ is divergent. On the other hand, we show that if
this series is convergent, then such an algorithm can be given.
This shows, in particular, that the best possible competitive
ratio is $O(k \log k (\log \log k)^2)$ (since the series for this
function converges according to the Cauchy condensation test), and
it is $\Omega(k\log k \log \log k)$ (since the series for this
function diverges according to the Cauchy condensation test). In
fact it is
$
O(k \log k \log\log k \cdots$
$(\log^{(x)} k)^2)$ and $\Omega(k \log k
\log\log k \cdots \log^{(x)} k)
$,
for any positive integer $x$.

Consider a sequence $f(i)$ for which $S_f$ is convergent, and let
$c_f$ be its limit. We present an algorithm, \batch, for this
variant of sum coloring. Initially, all colors are declared {\em
available}. When coloring the $i$th batch, its induced subgraph is
first colored using an optimal procedure, \optcolor. Let $t_i$
denote the maximum color used by \optcolor for batch $i$. For each
$j=1,2,\ldots,t_i$ in increasing order, vertices that \optcolor
gives color $j$ will be colored using the largest available color
among the colors $1,2,\ldots,\lfloor j \cdot c_f \cdot
f(i)\rfloor$. Then, this color is declared {\em taken}. This color
is now unavailable for vertices of future batches and for vertices
of the current batch that were assigned a color larger than $j$ by
\optcolor.
 If this process is
successful (there always exists an available color), then we say that
batch $i$ is {\it feasible}.

Assuming that all batches are feasible, using arguments similar to
those used for Lemma~\ref{knownk-ub}, we obtain an upper bound on
the competitive ratio of \batch as follows. Since a color used by
\optcolor in a particular batch is assigned to an available color
by \batch, if all batches are feasible, each pair, $(i,j)$, where
$i$ is a batch number and $j$ is a color assigned by \optcolor in
batch $i$, is given a different color. Since \optcolor produces a
proper coloring, \batch does too. The function $f$ is
nondecreasing, so the color assigned to a given vertex by \batch
is at most $c_f\cdot f(k)$ times the color assigned by \optcolor.

\begin{lemma}
\label{lemma:compr}
Consider sum coloring in $k$ batches, where the value of $k$ is not
known in advance.
If for all $1 \leq i \leq k $, batch $i$ is feasible, then the
competitive ratio of \batch is at most $c_f \cdot f(k)$.
\end{lemma}

\begin{lemma}
\label{lemma:feasible}
All batches for the algorithm \batch are feasible.
\end{lemma}
\begin{proof}
Assume that the algorithm has an infeasible batch, let $i$ be the
minimal index of a batch that is not feasible, and $j$ be the
smallest color that was used by \optcolor, for which \batch cannot
find an available color among the first $\lfloor j \cdot c_f
\cdot f(i) \rfloor$ colors. 
Let $t+1$ be the smallest available color at the time when \batch
tries to select a color for vertices that \optcolor gives color
$j$ in batch $i$. That is, all the $t$ smallest colors were
selected earlier (during the first $i-1$ batches or earlier during
batch $i$), and color $t+1$ is still available. By definition,
$t+1>\lfloor j \cdot c_f \cdot f(i)\rfloor$.

The color $t+1$ was available when previous colors were selected.
Consider a pair $j',\ell$ such that $\ell \leq i$, $1 \leq j' \leq
t_{\ell}$, and $j'<j$ if $\ell=i$. If $t+1 \leq \lfloor j' \cdot
c_f \cdot f(\ell)\rfloor$, then the color selected by \batch for
\optcolor's color $j'$  for batch $\ell$ is above $t+1$, since the
maximum available color no larger than $\lfloor j' \cdot c_f \cdot
f(\ell) \rfloor$ was selected. Thus, all colors $1,2,\ldots,t$
were selected for pairs $j',\ell$ satisfying $\lfloor j' \cdot c_f
\cdot f(\ell)\rfloor \leq t$, and thus $j' \cdot c_f \cdot
f(\ell)< t+1$. For a given value of $\ell$, the number of suitable
values of $j'$ is smaller than $\frac{t+1}{c_f\dot f(\ell)}$. As
the color $t+1$ cannot be selected for \optcolor's color $j$ for
batch $i$, $j$ is one such value for batch $i$, so for this batch
the number of values of $j'$ whose selected colors are no larger
than $t$ is smaller than $\frac{t+1}{c_f\dot f(i)}-1$. The total
number of colors strictly below $t+1$ selected in the first $i$
batches just before \optcolor's color $j$ for batch $i$ is
considered is strictly below $\sum_{\ell=1}^{i} \frac{t+1}{c_f\dot
f(\ell)} -1 \leq \frac{t+1}{c_f}\sum_{\ell=1}^{\infty} \frac
{1}{f(\ell)}-1=t$, where the last inequality holds since the
series converges to $c_f$, contradicting the assumption that all
the first $t$ colors were already selected.
\mbox{}\qed\end{proof}

By Lemmas~\ref{lemma:compr} and~\ref{lemma:feasible}, we obtain:
\begin{theorem}\label{unknownk-ub}
Consider sum coloring in at most $k$ batches and let $f$ be
any nondecreasing function with $f(i)\geq 1$ for all
$i\geq 1$, whose
 series $S_f$ converges to $c_f$.
Then, the algorithm \batch is $(c_f\cdot
f(k))$-competitive, even if the value $k$ is not known in advance.
\end{theorem}

Now, we provide the lower bound.

\begin{theorem}\label{unknownk-lb}
Consider sum coloring in $k$ batches, where the value of $k$ is not
known in advance.
Let $f(i)$ be a nondecreasing sequence with $f(i)\geq 1$ for all
$i\geq 1$,
 whose series $S_f$ is divergent.
Then, there is no constant $\ratio$ such that a
competitive ratio of at most $\ratio \cdot f(k)$ can be obtained
for all $k\geq 1$.
\end{theorem}
\begin{proof}
Assume for the sake of contradiction that there exists a constant
$c>1$ and an algorithm \alg, such that \alg is $(\ratio \cdot
f(k))$-competitive, for any number $k \geq 1$ of
batches. Let $\largerratio =\max\{2\ratio,10\}$. Let $k$ be such
that $\sum_{i=1}^{k} {1}/{f(i)}> 11 \largerratio$ (where
$k$ must exist as the series $S_f$ is divergent). Fix a large
integer $M$, such that $M > 130 \cdot \largerratio^2 \cdot
f(k)^2$. We say that a color $a$ is {\it small} if $a\leq 10
\largerratio M$.

We now describe an adversarial input. Batch $i$ of the input
consists of $M^{i-1}$ cliques of size
$3\lfloor{M}/{f(i)}\rfloor$. There are no edges between
vertices in different cliques of the same batch. A vertex that
\alg colors with a small color is called a {\em cheap} vertex. For
each batch $i$, if there is at least one clique containing at
least $M/f(i)$ cheap vertices, then one such clique is chosen, and
the cheap vertices of this clique are called {\em special}
vertices. In each batch, all vertices are connected to all special
vertices of previous batches and to no other vertices in previous
batches. Thus, no colors used for special vertices can be used in
later batches, and there is at most one special vertex for each
small color.

The input will contain at most $k$ batches. If, after some batch
$i<k$, the sum of colors used by \alg is larger than $c \cdot
f(i)$ times the optimal sum of colors, there will be no more
batches. Otherwise, all $k$ batches are given. Thus, if there are
fewer than $k$ batches, the theorem trivially follows. Below, we
consider the case where there are exactly $k$ batches.

We first give an upper bound on the optimal sum of colors for the
first $i$ batches, for $1 \leq i \leq k$.

{\em Claim}: For every value of $i$ (such that $1\leq i \leq k$), the optimal
sum of colors for the first $i$ batches is at most
$19M^{i+1}/(f(i))^2$.

We now prove the claim:
Consider the
following proper coloring. For each clique $K$, let $n_K$ denote
the number of vertices in $K$ that are not special. These vertices
are colored using the colors $1,2,\ldots,n_K$. Each special vertex
$v$ is given the color $3M+b$, where $b$ is the color assigned to
$v$ by \alg. As the vertices of each clique are only connected to
special vertices of previous batches, and they are not connected
to vertices of other cliques of the same batch, this coloring is
proper.

For $i=1$, there is only one clique, and the sum of colors in this
coloring (where there is one vertex of every color in
$\{1,2,\ldots3\floor{\frac{M}{f(i)}}\}$), is
$$\sum_{\ell=1}^{3\floor{\frac{M}{f(i)}}} \ell < \frac{9
  M^2}{(f(i))^2} < \frac{19 M^{i+1}}{(f(i))^2}\,.$$
For $i \geq 2$, the sum of the colors of special vertices is at
most
$$\sum_{\ell=3M+1}^{3M+10 \largerratio M} \ell < (10 \largerratio
M)(3M+10 \largerratio M) < 130 \largerratio^2 M^2 <
\frac{M^3}{(f(k))^2} \leq \frac{M^3}{(f(i))^2} \leq
\frac{M^{i+1}}{(f(i))^2}\,,$$ and the sum of the colors of the
remaining vertices is at most
$$\sum_{j=1}^{i} M^{j-1} \sum_{\ell=1}^{3\floor{\frac{M}{f(j)}}} \ell
  < \sum_{j=1}^{i} M^{j-1}\frac{9M^2}{(f(j))^2} = \sum_{j=1}^{i} \frac{9M^{j+1}}{(f(j))^2}
< \frac{18 M^{i+1}}{(f(i))^2} \,,$$ where the last inequality
follows by showing that for every $j\leq i$ we have
$\frac{9M^{j+1}}{(f(j))^2} \leq \frac{9M^{i+1}}{(f(i))^2} \cdot
\frac{1}{2^{i-j}}$ by induction on $i-j$.  For $i-j=0$, the claim
trivially holds.  Assume that it holds for $i-j$ and denote
$j'=j-1$, we will show it for $i-j'$. We have
$$\frac{9M^{j'+1}}{(f(j'))^2} = \frac{9M^{j+1}}{(f(j))^2} \cdot
\frac{(f(j))^2}{M\cdot (f(j'))^2} \leq \frac{9M^{j+1}}{(f(j))^2}
\cdot \frac{(f(i))^2}{M} \leq \frac{1}{2} \cdot
\frac{9M^{j+1}}{(f(j))^2} \leq \frac{9M^{i+1}}{(f(i))^2} \cdot
\frac{1}{2^{i-j'}}$$ where the first inequality holds because $1
\leq f(j') \leq f(j)  \leq f(i)$, the second inequality holds by
our choice of $M$, and the last inequality holds by the induction
assumption.
 Thus, for this coloring, the total sum of the colors is less
than $19M^{i+1}/(f(i))^2$.

This concludes the proof of the claim.

We now show that, by the assumption that \alg is $(c \cdot
f(i))$-competitive on $i$ batches, $1 \leq i \leq k$, each batch
$i$ must have a clique with at least $M/f(i)$ cheap vertices.
Assume for the sake of contradiction that some batch $i$ does not
contain a clique with at least $M/f(i)$ cheap vertices. Then, each
clique in the batch contains at most $\floor{M/f(i)}$ cheap
vertices and hence at least $2\floor{M/f(i)}$ vertices with colors
larger than $10 \largerratio M$. Thus, the sum of colors used for
this batch is more than $M^{i-1} \cdot 2\floor{M/f(i)} \cdot 10
\largerratio M > 10 \largerratio M^{i+1}/f(i) \geq 20 \ratio
M^{i+1}/f(i)$. By Claim 1
, this gives a
ratio of more than $$\frac{20 \ratio
M^{i+1}/f(i)}{19M^{i+1}/(f(i))^2} > \ratio \cdot f(i)\,.$$ Thus,
the total number of special vertices is at least $\sum_{i=1}^{k}
M/f(i) > 11 \largerratio M$, contradicting the fact that there is
at most one special vertex for each of the small colors.
\mbox{}\qed\end{proof}

\bibliographystyle{abbrv}
\bibliography{batch_color}

\end{document}

%% file: lemma4.pdf_t
\begin{picture}(0,0)%
\includegraphics{lemma4.pdf}%
\end{picture}%
\setlength{\unitlength}{1243sp}%
\begingroup\makeatletter\ifx\SetFigFont\undefined%
\gdef\SetFigFont#1#2#3#4#5{%
  \reset@font\fontsize{#1}{#2pt}%
  \fontfamily{#3}\fontseries{#4}\fontshape{#5}%
  \selectfont}%
\fi\endgroup%
\begin{picture}(10641,4804)(1318,-6740)
\put(3151,-6586){\makebox(0,0)[b]{\smash{{\SetFigFont{9}{10.8}{\rmdefault}{\mddefault}{\updefault}{\color[rgb]{0,0,0}$p_{\ell}$}%
}}}}
\put(9451,-6586){\makebox(0,0)[b]{\smash{{\SetFigFont{9}{10.8}{\rmdefault}{\mddefault}{\updefault}{\color[rgb]{0,0,0}$p_r$}%
}}}}
\put(10576,-2311){\makebox(0,0)[b]{\smash{{\SetFigFont{9}{10.8}{\rmdefault}{\mddefault}{\updefault}{\color[rgb]{0,0,0}$\mathcal{I}_r$}%
}}}}
\put(2026,-2311){\makebox(0,0)[b]{\smash{{\SetFigFont{9}{10.8}{\rmdefault}{\mddefault}{\updefault}{\color[rgb]{0,0,0}$\mathcal{I}_{\ell}$}%
}}}}
\put(2431,-3616){\makebox(0,0)[b]{\smash{{\SetFigFont{9}{10.8}{\rmdefault}{\mddefault}{\updefault}{\color[rgb]{0,0,0}$I_{\ell}$}%
}}}}
\put(10126,-3571){\makebox(0,0)[b]{\smash{{\SetFigFont{9}{10.8}{\rmdefault}{\mddefault}{\updefault}{\color[rgb]{0,0,0}$I_r$}%
}}}}
\end{picture}%

%% file: arxiv2.bbl
\begin{thebibliography}{10}

\bibitem{BBDGT15}
J.~Balogh, J.~B\'{e}k\'{e}si, G.~D\'{o}sa, G.~Galambos, and Z.~Tan.
\newblock Lower bound for 3-batched bin packing.
\newblock {\em Discrete Optimization}, 21:14--24, 2016.

\bibitem{BBGM09}
J.~Balogh, J.~B{\'{e}}k{\'{e}}si, G.~Galambos, and M.~C. Mark{\'{o}}t.
\newblock Improved lower bounds for semi-online bin packing problems.
\newblock {\em Computing}, 84(1--2):139--148, 2009.

\bibitem{BBHST98}
A.~Bar-Noy, M.~Bellare, M.~Halldorsson, H.~Shachnai, and T.~Tamir.
\newblock On chromatic sums and distributed resource allocation.
\newblock {\em Information and Computation}, 140:183--202, 1998.

\bibitem{Bor+}
A.~Borodin, I.~Ivan, Y.~Ye, and B.~Zimny.
\newblock On sum coloring and sum multi-coloring for restricted families of
  graphs.
\newblock {\em Theoretical Computer Science}, 418:1--13, 2012.

\bibitem{D15}
G.~D\'{o}sa.
\newblock Batched bin packing revisited.
\newblock {\em Journal of Scheduling}, 2015.
\newblock In press.

\bibitem{E16}
L.~Epstein.
\newblock More on batched bin packing.
\newblock {\em Operations Research Letters}, 44:273--277, 2016.

\bibitem{GLS81}
M.~Gr{\"o}tschel, L.~Lov{\'a}sz, and A.~Schrijver.
\newblock The ellipsoid method and its consequences in combinatorial
  optimization.
\newblock {\em Combinatorica}, 1(2):169--197, 1981.

\bibitem{GJY05}
G.~Gutin, T.~Jensen, and A.~Yeo.
\newblock Batched bin packing.
\newblock {\em Discrete Optimization}, 2:71--82, 2005.

\bibitem{GLff88}
A.~Gy\'{a}rf\'{a}s and J.~Lehel.
\newblock On-line and {First-Fit} colorings of graphs.
\newblock {\em Journal of Graph Theory}, 12:217--227, 1988.

\bibitem{GL90}
A.~Gy\'{a}rf\'{a}s and J.~Lehel.
\newblock First fit and on-line chromatic number of families of graphs.
\newblock {\em Ars Combinatoria}, 29(C):168--176, 1990.

\bibitem{H00}
M.~M. Halld{\'{o}}rsson.
\newblock Online coloring known graphs.
\newblock {\em The Electronic Journal of Combinatorics}, 7, 2000.

\bibitem{HS94}
M.~M. Halld{\'{o}}rsson and M.~Szegedy.
\newblock Lower bounds for on-line graph coloring.
\newblock {\em Theoretical Computer Science}, 130(1):163--174, 1994.

\bibitem{KST16}
H.~A. Kierstead, D.~A. Smith, and W.~T. Trotter.
\newblock First-fit coloring on interval graphs has performance ratio at least
  5.
\newblock {\em European Journal of Combinatorics}, 51:236--254, 2016.

\bibitem{KT81}
H.~A. Kierstead and W.~T. Trotter.
\newblock An extremal problem in recursive combinatorics.
\newblock {\em Congressus Numerantium}, 33:143--153, 1981.

\bibitem{kub04}
E.~Kubicka.
\newblock The chromatic sum of a graph: History and recent developments.
\newblock {\em International Journal of Mathematics and Mathematical Sciences},
  2004(30):1563--1573, 2004.

\bibitem{KS89}
E.~Kubicka and A.~J. Schwenk.
\newblock An introduction to chromatic sums.
\newblock In {\em 17th ACM Computer Science Conference}, pages 39--45. ACM
  Press, 1989.

\bibitem{NB08}
N.~S. Narayanaswamy and R.~S. Babu.
\newblock A note on {First-Fit} coloring of interval graphs.
\newblock {\em Order}, 25(1):49--53, 2008.

\bibitem{NSS99}
S.~Nicolosoi, M.~Sarrafzadeh, and X.~Song.
\newblock On the sum coloring problem on interval graphs.
\newblock {\em Algorithmica}, 23(2):109--126, 1999.

\bibitem{ZCW03}
G.~Zhang, X.~Cai, and C.~K. Wong.
\newblock Scheduling two groups of jobs with incomplete information.
\newblock {\em Journal of Systems Science and Systems Engineering}, 12:73--81,
  2003.

\end{thebibliography}
